\documentclass{article}

\usepackage[left=1in,right=1in,top=1in,bottom=1in]{geometry}

\usepackage{dsfont}             %
\usepackage{amsmath}
\usepackage{amsthm}
\usepackage{amssymb}
\usepackage{mathtools}
\usepackage{algorithmicx}
\usepackage{algorithm}
\usepackage{graphicx}
\usepackage{color}
\usepackage{subcaption}
\usepackage{xcolor}
\usepackage{forest}
\usepackage{booktabs}

\usepackage{hyperref}

\usepackage{mathabx}

\theoremstyle{plain}% default
\newtheorem{theorem}{Theorem}[section]
\newtheorem{lemma}[theorem]{Lemma}
\newtheorem{definition}[theorem]{Definition}
\newtheorem{observation}[theorem]{Observation}

\newtheorem{proposition}[theorem]{Proposition}

\theoremstyle{remark}

\newtheorem*{note}{Note}

\newtheorem*{mainq}{Main Question}

\newcommand{\eps}{\varepsilon}

\newtheorem*{ksumconj}{The Strong $k$SUM Conjecture}

\newcommand{\diag}{\mathrm{diag}}

\begin{document}

%%%%%%%%%%%%%%%%%%%%%%%%%%%%%%%%%%%%%%%%%%%%%%%%%%%%%%%%%%%%%%%%%%%%%%%%%%%%%%%%%
%
%                         TITLE
% 
%%%%%%%%%%%%%%%%%%%%%%%%%%%%%%%%%%%%%%%%%%%%%%%%%%%%%%%%%%%%%%%%%%%%%%%%%%%%%%%%%

%% \icmltitlerunning{Subquadratic High-Dimensional Hierarchical Clustering}
%% \icmltitle{Subquadratic High-Dimensional Hierarchical Clustering}

\title{Impossibility Results for Grammar-Compressed Linear Algebra}

\date{}

\author{%
  Amir Abboud\thanks{IBM Almaden Research Center, \texttt{amir.abboud@gmail.com}}
  \and
  Arturs Backurs\thanks{Toyota Technological Institute at Chicago, \texttt{backurs@ttic.edu}. Supported by an NSF Grant CCF-2006806.}
  \and
  Karl Bringmann\thanks{Saarland University and Max Planck Institute for Informatics, Saarland Informatics Campus, \texttt{bringmann} \texttt{@cs.uni-saarland.de}. This work is part of the project TIPEA that has received funding from the European Research Council (ERC) under the European Unions Horizon 2020 research and innovation programme (grant agreement No.\ 850979).}
  \and
  Marvin K\"unnemann\thanks{Max Planck Institute for Informatics, Saarland Informatics Campus, \texttt{marvin@mpi-inf.mpg.de}}
}

\maketitle
\begin{abstract}
  
To handle vast amounts of data, it is natural and popular to compress vectors and matrices. 
When we compress a vector from size $N$ down to size $n \ll N$, it certainly makes it easier to store and transmit efficiently, but does it also make it easier to process?

In this paper we consider lossless compression schemes, and ask if we can run our computations on the compressed data as efficiently as if the original data was that small. That is, if an operation has time complexity $T(\text{input-size})$, can we perform it on the compressed representation in time $T(n)$ rather than $T(N)$?
We consider the most basic linear algebra operations: inner product, matrix-vector multiplication, and matrix multiplication.
In particular, given two compressed vectors, can we compute their inner product in time $O(n)$? Or perhaps we must decompress first and then multiply, spending $\Omega(N)$ time?

The answer depends on the compression scheme. While for simple ones such as Run-Length-Encoding (RLE) the inner product can be done in $O(n)$ time, we prove that this is impossible for compressions from a richer class: essentially $n^2$ or even larger runtimes are needed in the worst case (under complexity assumptions).
This is the class of \emph{grammar-compressions} containing most popular methods such as the Lempel-Ziv family.
These schemes are more compressing than the simple RLE, but alas, we prove that performing computations on them is much harder.

\end{abstract}

\section{Introduction}
\label{sec:intro}

% !TEX root = main.tex

The idea of using compression to speed up computations can be found in any domain that deals with large-scale data, and ML is no exception.
By exploiting redundancies and various forms of structure in a piece of data, compression algorithms such as zip can reduce its size from $N$ down to $n$, where $n  \ll N$.
The data becomes cheaper to store, access, transmit, and perhaps also to analyze. Can we run our ML tools on the compressed data, without decompressing it first, and make the computation times proportional to $n$ rather than $N$?
Since most ML algorithms boil down to large amounts of basic algebraic operations such as multiplications of vectors and matrices, with \emph{inner product} as the atomic operation, the most basic question in this context is:

\begin{mainq}
Given two $N$-dimensional vectors, each in a compressed form of size $n \ll N$, can we compute their inner product in $\tilde{O}(n)$ time\footnote{We use the notation $\tilde{O}(n)=n\cdot N^{o(1)}$ for near-linear time, hiding small terms such as log factors.} rather than $O(N)$?
\end{mainq}

The answer, of course, depends on the compression scheme that we use. 
There seems to be an inherent tension: more complex schemes have higher compression rates but are harder to analyze without decompression. 

First, let us clarify that our interest is in exact computations and \emph{lossless} compressions, 
even though lossy techniques such as dimensionality reduction  \cite{lossy01} are widely used by the ML community.
In many cases, e.g. when performing a basic algebraic operation within a larger pipeline, even a small amount of error could add up to make the final result unintelligible.
Recent years has seen a growing interest in exploring the potential of lossless compression for speeding up ML \cite{CLA_IBM_CACM19,tabei2016scalable,LZW_ML17,MJF19}.
An inspiring result was honorably mentioned as an outstanding paper at NeurIPS last year \cite{MJF19}: any $N \times d$ matrix $A$ can be compressed down to a matrix of size $d \times d$ such that the optimal solutions of Least-Mean-Squares (LMS) instances are exactly the same on $A$ and $A'$.
This is an example where for a specific task (LMS solvers) a specific compression scheme (designed by the authors) leads to a solution in time $T(n)$ rather than $T(N)$,  giving a 100x speedup on benchmark data; it makes one wonder if this approach can work in a more general setting.

For rather simple compression methods, the answer to our question is positive. 
A recent Communications of the ACM article \cite{CLA_IBM_CACM19} exhibits \emph{Compressed Linear Algebra} \cite{CLA_IBM16,CLA_IBM17,CLA_IBM18} a compression scheme for vectors and matrices that uses simple techniques such as Run Length Encoding (RLE) and allows for fast computations on the compressed data with impressive experimental results when integrated into ML systems.
The RLE encoding of a vector simply replaces runs of values by tuples indicating the value and the length of the run; e.g. the binary vector $00011111000$ gets encoded as $0^31^50^3$. Given two vectors encoded in this way with size $n_{RLE}$, a simple one-pass algorithm can compute their inner product in $O(n_{RLE})$ time. 
Before that, there were many algorithms for exploiting succinct encodings of \emph{sparse} vectors \cite{Sparse_book03,sparse_Mv_compression08,sparse_Mv_compression12}; e.g. by simply listing the nonzero locations the binary vector $0100001000$ gets encoded as $(2,7)$.
These encodings allow for a linear time inner product computation as well.

However, these simple methods are often not very compressing. At the other end of the spectrum, we have the heavy-duty and time-tested family of \emph{Grammar-Compressions}~\cite{KiefferY00} that includes the Lempel-Ziv-family (LZ77, LZ78, LZW, etc.) \cite{LZ76,LZ77,W84}, Byte-Pair Encoding \cite{BytePair}, dictionary methods, and others \cite{NW97,Liu+08}.
These compressions are used in ubiquitous applications such as zip, Snappy, GIF, PNG, the built-in Unix utility {\tt compress}, and even in PDF.
Their compression rates are often on a whole different level compared to RLE; e.g. the current draft of this paper reduces from 10KB to 4KB with zip but RLE has no effect. 
See Table~\ref{tab:compress-rates} and \cite[Table 1]{CLA_IBM_CACM19} for empirical data showing the quantitative potential of these methods for some standard ML datasets. 
What all these more elaborate compression techniques have in common is that they essentially (up to low order terms \cite{Rytter03}) encode a string (or vector) by a \emph{Straight-Line Program} (SLP): a restricted kind of a context-free grammar that can only produce one string.
In more detail, an SLP is defined over some alphabet $\Sigma$, say $\{0,1\}$, and it is a set of replacement rules (or productions) of a very simple form: a rule is either a symbol in $\Sigma$ or it is the concatenation of two previous rules  (under some fixed ordering of the rules). The last replacement rule is the sequence defined by the SLP.
For example, we can compress the sequence $01010101$ with the rules $S_1\to 0; \ \ S_2\to 1; \ \ S_3\to S_1\,S_2; \ \ S_4\to S_3\,S_3; \ \ S_5\to S_4\,S_4 \ $ and $S_5$ corresponds to the sequence $01010101$. 
See Figure~\ref{fig_slp}.
For some strings this can give an exponential compression, e.g. the sequence $(01)^N$ requires only $O(\log{N})$ rules; note that its RLE has size $N$.
While finding the smallest SLP for a given string is NP-Hard, it can be approximated either by the above practical methods or provably up to logarithmic factors \cite{Rytter03,Char+05,Sak05,Jez15,Jez16}.

%Indeed, the potential of grammar-compressions over RLE is observable on real world datasets (see Table~\ref{tab:compress-rates} for an overview of compression rates). For some datasets, such as ISOLET~\cite{...}, the compression rate of RLE is very limited, while zip reduces the size by a factor of $??$ to $??$.  

\begin{table}
  \caption{The potential savings from grammar-compressed linear algebra: Compression rates on real datasets. We compare zip, a standard grammar-compression, with Run Length Encoding (RLE), a simple method that works well on repetitive or sparse data. For more such results, see \cite[Table 1]{CLA_IBM_CACM19}.}
  \label{tab:compress-rates}
  \centering
  \begin{tabular}{lccc}
    \toprule
    Dataset     & Size     & RLE (compression rate) & zip (compression rate) \\
    \midrule
    ISOLET~\cite{UCI} &  30.94 MB & 29.83 MB (0.96) & 7.94 MB (0.26) \\
    US Census 1990~\cite{UCI} &  342.26 MB & 341.97 MB (0.99) & 51.91 MB (0.15) \\
    \bottomrule
  \end{tabular}
\end{table}

Thus, the holy grail in this context is to perform algebraic operations in $T(\text{compression-size})$ time \emph{even when} the vectors are compressed with zip or one of the other heavy-duty grammar compressions; that is, without unzipping them first.
Ideally, we would implement a ``zip-inner-product'' function that takes two zip files encoding vectors and computes the inner product in near-linear time (which may not even be enough time to unzip them).
A recent paper titled ``When LZW meets ML'' \cite{LZW_ML17} makes partial progress towards this goal: the inner product can be computed efficiently on their \emph{tuple oriented coding} where each \emph{coordinate} is grammar-compressed separately, but not the vector as a whole. 
This makes their method less compressing since, unlike with zip, the size of the encoding is always at least the dimensionality of the vectors.

\begin{mainq}[Restated]
Given two $N$-dimensional vectors, each grammar-compressed down to size $n \ll N$, can we compute their inner product in $\tilde{O}(n)$ time rather than $O(N)$?
\end{mainq}

While efficiently analyzing these grammars may seem like a daunting task, a large body of works over the last three decades has equipped us with an ingenious toolbox exactly for this purpose.
It turns out that many important problems can indeed be solved surprisingly faster than the decompress-then-solve bound, e.g. in pattern matching  \cite{Pla94,KRS95,ABF96,FT95,CGLM06,LMWZ09,Gaw11,HLLW13,Jez15pm}.
This gives hope for a positive answer to our question and that many ML computations could be sped up by operating on grammar-compressions.
These algorithms typically look at the parse trees that have $N$ leaves but only $n$ distinctly labelled internal nodes (see Figure~\ref{fig_slp}), and traverse them starting from the root down, while attempting to only spend time proportional to the depth of the tree per distinct label. 
Using tricks that restructure the grammar to make the tree balanced, the depth can be upper bounded by $O(\log N)$, making the total time $O(n \log N)$.
To learn more about this subfield of Algorithm Design, we refer the reader to the surveys \cite{WMB99book,Lar99book,GKPR96survey,SB06survey,GSU09survey,RB10survey,Rytt04survey,Lohrey12,Sak14survey}.

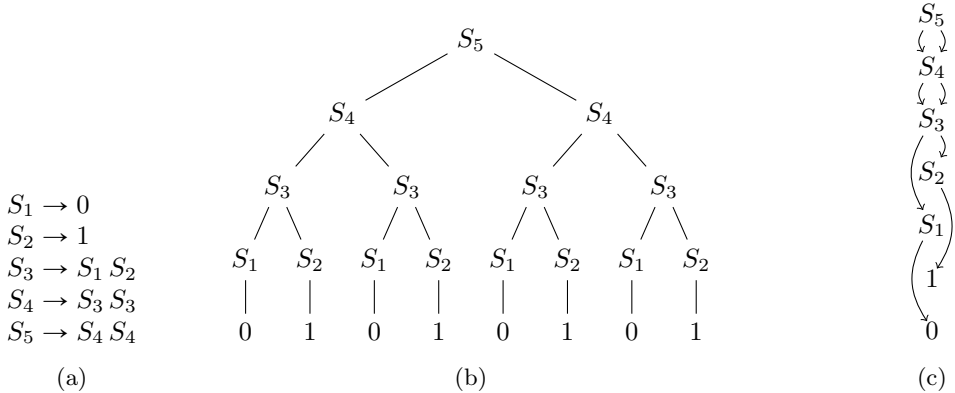
\begin{figure*}
	\centering
	
	\begin{subfigure}[b]{0.2\textwidth}
        \centering
			\begin{tabular}{ l }
			 $S_1\to 0$ \\
			 $S_2\to 1$ \\ 
			 $S_3\to S_1\, S_2$ \\ 
			 $S_4\to S_3\, S_3$ \\ 
			 $S_5\to S_4\, S_4$
			\end{tabular}
        \caption{}
        \label{fig_a}
    \end{subfigure}
	~
	\begin{subfigure}[b]{0.4\textwidth}
        \centering
        \begin{forest}
			[$S_5$
				[$S_4$[$S_3$[$S_1$[$0$]][$S_2$[$1$]]][$S_3$[$S_1$[$0$]][$S_2$[$1$]]]]
				[$S_4$[$S_3$[$S_1$[$0$]][$S_2$[$1$]]][$S_3$[$S_1$[$0$]][$S_2$[$1$]]]]
			]
		\end{forest}
        \caption{}
        \label{fig_b}
    \end{subfigure}
	~
	\begin{subfigure}[b]{0.3\textwidth}
        \centering
		%[shorten >=1pt,auto,node distance=1cm]
        \begin{tikzpicture}[
      %>=Stealth,
      shorten >=-3pt,
      shorten <=-3pt,
      auto,
      node distance=.7 cm,
      scale = 1,
      transform shape,
      state/.style={circle,inner sep=2pt}
      ]
			%\node[initial,state] (q1)      {$1$};
			\node[state] (z0) {$0$};
			\node[state] (z1) [above of=z0] {$1$};
			\node[state] (s1) [above of=z1] {$S_1$};
			\node[state] (s2) [above of=s1] {$S_2$};
			\node[state] (s3) [above of=s2] {$S_3$};
			\node[state] (s4) [above of=s3] {$S_4$};
			\node[state] (s5) [above of=s4] {$S_5$};
			%\node[state, accepting]         (q3) [above right of=q4] {$3$};

			%\path[->]          (q1)  edge                 node {a} (q2);
			\path[->] (s1) edge [bend right] (z0);
			\path[->] (s2) edge [bend left] (z1);
			\path[->] (s3) edge [bend right=30] (s1);
			\path[->] (s3) edge [bend left] (s2);
			\path[->] (s4) edge [bend left] (s3);
			\path[->] (s4) edge [bend right] (s3);
			\path[->] (s5) edge [bend left] (s4);
			\path[->] (s5) edge [bend right] (s4);
			%\path[->]          (q3)  edge   [bend left]   node {a} (q2);
			%\path[->]          (q2)  edge                 node {b} (q4);
			%\path[->]          (q4)  edge                 node {a} (q3);
		\end{tikzpicture}
        \caption{}
        \label{fig_c}
    \end{subfigure}
	
	\caption{(a) An SLP generating the sequence $01010101$. (b) The corresponding parse tree. (c) The acyclic graph corresponding to the SLP.}
	\label{fig_slp}
\end{figure*}

\subsection{Our Results}

Alas, our main result is a negative resolution to the main question above.
We apply the tools of theoretical computer science, and the recently blossoming field of \emph{fine-grained complexity}, in order to shed light into the mathematical foundations of Compressed Linear Algebra.
We prove new hardness reductions showing cases where the time to compute the inner product must be large (under popular complexity assumptions) even when the vectors have very small grammar compressions. 
For example, there are $N$-dimensional vectors with grammar-compressions of size $n=O(N^{1/3})$ where the inner product must take $\tilde{\Omega}(n^2)$ time\footnote{The more standard notation is $n^{2-o(1)}$ which indicates an $\Omega(n^{1.9999})$ lower bound, no matter how close to~$2$ we go. That is, only mildly subquadratic algorithms are possible, e.g. by shaving log factors.} to compute.
The consequences to other settings such as matrix-vector multiplication are further explained below.
This creates a strong separation between grammar-compressions, where we prove an $\tilde{\Omega}(n^2)$ lower bound, and RLE, where an $O(n)$ algorithm exists. This formally justifies the use of simpler methods in ML systems and guides researchers away from searching for an ultra-efficient ``zip-inner-product'' function.

\vspace{-0.2cm}

\paragraph{Fine-Grained Complexity} Negative results are paramount to the success of any scientific discipline. 
The most prominent framework for proving such results in computer science is the theory of NP-Hardness, where one proves that a problem cannot be solved in polynomial time unless $P = NP$ which would imply breakthrough algorithms for famously-hard problems such as SAT and Subset Sum.
Without this theory, countless hours would have been wasted by algorithm designers trying to come up with provable, worst-case, polynomial time algorithms for NP-Hard problems.
Due to the increase in data sizes of recent years, the ethos of this theory that ``efficient = polynomial'' has become obsolete, and a more demanding attitude where ``efficient = linear'' has arisen.
By replacing the polynomial reductions of NP-Hardness with more efficient ones (often linear), fine-grained complexity can prove hardness results even for problems that have polynomial time algorithms. 
Exemplary results show that linear or subquadratic
algorithms for certain problems, which admit quadratic-time algorithms,
would refute popular assumptions (conjectures that are similar to but stronger than $P \neq NP$) and have breakthrough consequences for famously hard problems.
One of the central assumptions in this theory and in this paper is the 3SUM Conjecture: ``\emph{No algorithm can decide, in subquadratic $O(n^{2-\eps})$ time, if there are three numbers that sum to zero among a given set of $n$ numbers}''.
A recent survey on fine-grained complexity \cite{Williams18_ICM} cites dozens of papers, mainly in computational geometry \cite{GO95} but also in other fields \cite{Pat10,vassilevska2009finding,abboud2014popular,abboud2014consequences,chen2009approximate,kopelowitz2016higher,jumbled3sum,goldstein2017hard}, that prove 3SUM-Hardness results showing that their algorithms are optimal up to a refutation of this conjecture. 
In this paper, we prove the first 3SUM-Hardness results in ML\footnote{We remark that \emph{some} complexity assumption is necessary for proving the kind of results we are interested, since unconditionally proving even very weak lower bounds on the time complexity such as $\Omega(n^{1+\eps})$ and even for NP-Hard problems like SAT (not to mention inner product) is far beyond current techniques \cite{ABbook}.} as far as we are aware.
The 3SUM assumption and its generalizations that we use in the theorems below are formally defined and discussed in Section~\ref{sec:prelim}. 

\vspace{-0.2cm}

\paragraph{Vector Inner Product}
Our first and main result is a reduction from 3SUM to compressed inner product of two vectors, negatively resolving our main question.

\begin{theorem}\label{thm:vectorInnerProduct}
Assuming the 3SUM conjecture, the inner
product of two $N$-dimensional vectors that are grammar-compressed to size $n=\Theta(N^{\frac{1}{4}})$
cannot be computed in $O(n^{2-\eps})$ time where $\eps>0$.
%Moreover, under the Strong 3SUM conjecture, the same lower bound holds when $n=\Theta(N^{\frac{1}{3}})$.
\end{theorem}

Moreover, we strengthen and generalize this result in several ways. 
First, we address the dependence between $n$ and $N$: could it be that for more or less compressed vectors the picture is different?
Using a stronger variant of the 3SUM conjecture, the same lower bound of $n^2$ holds even when $n=N^{1/3}$, and therefore our result can be stated as an $\tilde{\Omega}(N^{\frac{2}{3}})$ lower bound which is quite close to the trivial upper bound of $O(N)$. 
Moreover, by a (highly nontrivial) boosting of our reduction, in Section~\ref{sec:VV} we establish an $\tilde{\Omega}(N^{\frac{1}{3}})$ lower bound with $n=N^{\eps}$ for any $\eps \leq 1/3$. That is, when the vectors are highly compressed even $n^{10}$ time is not sufficient\footnote{Strictly speaking, such a conditional lower bound of $\Omega(n^{10})$ for \emph{highly compressible} inputs can already be proven by combining a known \#P-hardness reduction from SubsetSum~\cite{Lifshits07} with a fine-grained hardness of SubsetSum under the Exponential Time Hypothesis (see, e.g.~\cite{JansenLL16}). However, such an approach yields only a weak lower bound in terms of the uncompressed size $N$, namely a bound of $\Omega(N^\epsilon)$ for some non-explicit, possibly tiny $\epsilon$. Our lower bounds always give an explicit, reasonably large value for $\epsilon$.}; this is in stark contrast to the case of RLE-compressed vectors where $O(n)$ is always possible.

\vspace{-0.2cm}

\paragraph{Matrix-Vector Multiplication}
Next, we consider the problem of computing the $M \cdot v$ product of an $N$-dimensional vector $v$ that is compressed to size $n$ with an $N \times N$ matrix $M$ where each row is compressed to size $O(n)$.
Perhaps computing these $N$ inner products as a batch can be done faster than computing each separately.
Alas, by another significant boosting of our reduction we prove that this is also impossible.
While if the encoding is with RLE the product can be computed in $O(Nn)$ time, which is linear in the representation size of the matrix and thus optimal, it turns out that for grammar compressions $\tilde{\Omega}(Nn^2)$ is required. The proof is in Section~\ref{sec:MV}.

\begin{theorem}
\label{thm:MV}
Assuming the 3SUM conjecture, the product of
an $N\times N$-dimensional matrix, where each row is grammar-compressed
to size $n=\Theta(N^{\frac{1}{5}})$, with an $N$-dimensional vector that is grammar-compressed to size $n$ cannot be computed 
in $O(Nn^{2-\eps})$ time where $\eps>0$.\end{theorem}

\vspace{-0.2cm}

\paragraph{Matrix Multiplication} Finally, we consider matrix multiplication of compressed matrices $C=A \cdot B$. 
There are multiple ways to compress an $N\times N$ matrix: we might compress each row or each column, so that the compression size is $N\cdot n$, or treat the whole matrix as an $N^2$-dimensional vector and compress it to size $n$. 
Each way may lead to a different time complexity, but no matter which way we choose, the first question to ask, and that will determine the time we can hope for, is: what is the output size?
The na\"{i}ve answer is that the matrix $C$ has size $N \times N$, but since $A$ and $B$ are compressed, shouldn't we expect $C$ to also be  representable with a small grammar of size $n \ll N^2$?
Unlike the above questions that deal with computation time, this is an information-theoretic question, and in Section~\ref{sec:MM} we give strong and unconditional negative answers: the matrix $C$ cannot be grammar-compressed to size $o(N^2/\log^2{N})$ \emph{even} when $A$ and $B$ are strongly compressible.
Moreover, some of our results hold even when $A$ and $B$ have very small RLE encodings.
Therefore, our results should be of interest to the compressed linear algebra project beyond grammar-compressions.

\vspace{-0.2cm}

\paragraph{Technical Remarks}
While the tools for proving NP-Hardness results for grammar-compressed data are old \cite{Lohrey12}, they only apply in the unrealistic setting where $n = \log{N}$, and we are interested in more fine-grained results.
Only recently, a FOCS paper by the authors \cite{ABBK17} introduced the techniques for proving such lower bounds.
This previous work focused on combinatorial pattern matching problems and the current work extends it to the setting of linear algebra.
Our results establish the hardness even of the simplest setting of binary vectors and matrices over $\{0,1\}$. 
This setting is particularly studied due to its connection to graphs, where grammar compressions have also received a lot of attention \cite{maneth2015survey,maneth2018grammar}.
Moreover, we show that even deciding if the inner product is $0$ or $\geq 1$ is hard, and so our lower bounds hold against any bounded approximation algorithms. 
Extending the lower bounds to other functions such as computing the $\ell_2$ distance between two vectors is also easy.
Like almost all results in fine-grained complexity \cite{Williams18_ICM}, our lower bounds are against both deterministic and randomized algorithms. 

Finally, we remark that our lower bounds are for the most basic setting of \emph{worst-case} instances.
Extending them to \emph{average-case} results, showing that instances that come from certain natural distributions are also hard, is an open question.
However, notice that even if the original vectors come from a natural distribution, the distribution of the grammar representations will be completely different (and probably far from natural).
Therefore, exploiting the structure of non-worst-case instances seems far beyond current reach in this context.

\subsection{Other Related Works}
There have been a few recent works showing fine-grained complexity results for machine learning problems. In particular,~\cite{backurs2017improving} showed that the classic algorithm of Viterbi that computes the most likely path in a Hidden Markov Model which results in a given sequence of observations is essentially optimal assuming certain complexity theoretical hypotheses.
%the All-Pairs Shortest Paths and Min-Weight $k$-Clique Hypotheses.
Another work~\cite{backurs2017fine} showed conditional hardness results for multiple empirical risk minimization problems such as kernel support vector machines, kernel ridge regression, and training the final layer of a neural network.
Furthermore, there are many works that show hardness for problems that are used in machine learning literature. This includes conditional lower bounds for kernel low-rank approximation~\cite{musco2017input}, closest pair and its variants~\cite{alman2015probabilistic, rubinstein2018hardness, williams2018difference, chen2019equivalence, david2019complexity, cs2018closest}, maximum inner product~\cite{abboud2017distributed, chen2018hardness, chen2019fine}, earth mover's distance (a.k.a.\ Wasserstein  metric)~\cite{rohatgi2019conditional}, dynamic time warping distance~\cite{abboud2015tight, bringmann2015quadratic}.
 
%[This work shows how to estimate how much a compression scheme will save on the given data \cite{estimation_of_compression_rate}.]
Further contexts in which lossless compressions are used for ML applications, where the primary focus is on other aspects than increasing algorithmic performance, include compressing and accelerating models for deployment on resource-constrained devices (see~\cite{HanMD16,ChoudharyMGS20}; e.g., lossless compressions are used to compress weights after a quantization step) or implementing the principle of minimum description length for feature learning (see~\cite{PaskovWMH13}).

Outside of ML, the idea of improving efficiency by operating on (losslessly) compressed data is well-established in databases 
\cite{RDBMS_compression06,RDBMS_compression01,RDBMS_compression00,RDBMS_compression94}, and is gaining traction also in bioinformatics \cite{GTRAC16}.

\newtheorem*{3sumconj}{The 3SUM Conjecture}

\section{Preliminaries}
\label{sec:prelim}

As described in Section~\ref{sec:intro}, a grammar compression of a sequence (or a vector) is an SLP that produces the sequence. In our proofs we will use the following simple observation about SLPs.

\begin{proposition}\label{prop:SLPrepetition}
Let $\mathcal{G}$ be an SLP with start symbol $S$ that generates a sequence $s$. For any $\alpha \in \mathbb{N}$, we can compute an SLP $\mathcal{G}'$ that generates the $\alpha$-fold repetition of $s$, i.e., \[s^\alpha = \underbrace{s\;s\;\cdots\;s}_{\alpha \text{ times}},\] and has size $|\mathcal{G}| + O(\log \alpha)$ in time $O(|\mathcal{G}'|)$.
\end{proposition}
\begin{proof}[Proof sketch]
Using $O(\log \alpha)$ repeated squaring rules $S_i \rightarrow S_{i-1}S_{i-1}$ and $S_0 \rightarrow S$, we obtain non-terminals $S_0, \dots, S_{\lfloor \log_2 \alpha\rfloor}$ generating $s^{2^i}$ for $i\in \{0,\dots,\lfloor \log_2 \alpha \rfloor\}$. It is straightforward to combine these non-terminals, according to the binary representation of $\alpha$, to generate $s^{\alpha}$ using only $O(\log \alpha)$ additional non-terminals.
\end{proof}

\vspace{-0.2cm}

Using this property, we can often compress sequences much more efficiently than run-length encoding alone could: E.g., repetitive patterns like $(010011)^{n}$ can be encoded using only $\Theta(\log n)$ bits instead of $\Theta(n)$. Indeed, our constructions crucially exploit a repeated application of this property to compress hard instances to very small sizes.

\vspace{-0.2cm}

\paragraph{The Complexity Assumptions}
As discussed in Section~\ref{sec:intro}, the impossibility results in fine-grained complexity are based on certain popular conjectures. One of the central ones concerns the 3SUM problem, which has a few equivalent formulations (up to linear time transformations \cite{3LDT}); we will mostly use the following\footnote{For example, instead of $a+b=c$ or  $a+b+c=0$ we may be given a target $t$ and ask for $a+b+c=t$.}.

\begin{definition}[The 3SUM Problem]
Given three sets $A,B,C$ of $m$ integers in $\{1,\ldots, U\}$, decide if there is a triple $a \in A, b \in B, c \in C$ such that $a+b=c$.
\end{definition}

It is a simple exercise (that is often given in interviews) to come up with an $O(m^2)$ time algorithm, and despite decades of efforts, only mildly subquadratic $O(m^2/\log^c m)$ bounds for a small $0<c<3$ are known \cite{BDP05,GP14,GS15,freund2017improved,chan2019more}.

\begin{3sumconj}
No algorithm can solve the 3SUM problem in $O(m^{2-\eps})$ time, where $\eps>0$.
\end{3sumconj}

A few remarks about this conjecture.
First, a folklore trick of taking all numbers modulo a random large prime shows that the problem for arbitrary universe $U$ is equivalent to the case where $U=O(m^{3} \log^2 m)$  (see Lemma B.1 in \cite{abboud2014losing} for a proof). 
Therefore, we will assume this bound on $U$.
When $U$ becomes too small, the problem becomes easy due to an $O(m+U \log U)$ algorithm using Fast Fourier Transform \cite{CLRSbook}.
However, the problem is conjectured to be hard even when $U=\Theta(m^2)$ and this is referred to as the Strong 3SUM Conjecture \cite{jumbled3sum,ABBK17}. This stronger assumption allows us to strengthen our lower bounds by reducing $N$.
Second, the hardness of the more general kSUM problem is also used as a complexity assumption \cite{abboud2013exact,ABBK17}. 
In the formulation that we will use, we are given $k$ sets $A_1,\ldots,A_k$ of $m$ integers in $\{1,\ldots,U\}$ where $U=\Theta(m^{\lceil k/2 \rceil})$ and are asked to decide if there are $k$ numbers, one from each set, such that $a_1+\cdots+a_{k-1}=a_k$. The Strong kSUM conjecture states that cannot be done in  $O(m^{\lceil k/2 \rceil-\eps})$ time, for any $\eps>0$. We will use this assumption to prove lower bounds even when $n$ is much smaller than $N$.
Third, 3SUM and the other hardness assumptions in fine-grained complexity are conjectured to be true even against randomized algorithms that succeed with high probability. This is important since some of our reductions are randomized.

\section{Vector Inner Product}
\label{sec:VV}

In this section we present the proof of Theorem~\ref{thm:vectorInnerProduct} by giving a reduction from 3SUM to the inner product of compressed vectors.
A slightly weaker conditional lower bound of $\tilde{\Omega}(N^{1/2})$ for vectors compressible to $n=N^{1/4}$ can be extracted from the proof of Theorem 5.11 in \cite{ABBK17}. 
We use similar tricks, but a different and more optimized construction to obtain a stronger conditional lower bound of $\tilde{\Omega}(N^{2/3})$ already on less compressible vectors with $n=N^{1/3}$. Technically, the novelty is that we manage to encode two sets ($A$ and $B$) into one vector of length $mU$ rather than $m^2U$.
This new construction is crucial for the extensions we show -- we do not see how to prove any lower bound for matrix-vector inner product without building on this new construction.  

\begin{proof}
Given an instance of 3SUM, that is, three sets $A,B,C$ of $m$ integers in $\{1,\dots,U\}$, we show
how to construct vectors $v{}_{A+B}',v_{C}'\in\{0,1\}^{N}$ with $N=2mU\log^{2}m$
such that: (1) $v_{A+B}'\cdot v_{C}'\ge1$ if and only there $a\in A,b\in B,c\in C$
with $a+b=c$, (2) both vectors have a compression of size $O(m\log U)$, and (3) the construction time is $O(m\log U)$.

This reduction suffices for proving Theorem~\ref{thm:vectorInnerProduct} due to the following calculations.
Since (as discussed in Section~\ref{sec:prelim}) we can assume that $U=\Theta(m^{3}\log^2 m)$, the reduction produces two vectors of dimension $N=\Theta((m\log m)^{4})$
and compressed size $n=\Theta(N^{1/4})=\Theta(m\log m)$, such that the inner product reveals the answer to the 3SUM instance. 
Therefore, an $O(n^{2-\eps})$-time algorithm would solve the 3SUM instance in time $O(m^{2-\varepsilon}\mathrm{polylog}m)$, refuting the 3SUM conjecture.
Note that the $O(m\log U)$ time for the reduction itself is negligible.
Moreover, if we assume the Strong 3SUM conjecture, we can start with 3SUM instances where $U=O(m^{2})$ and get vectors of dimension $N=O((m\log m)^{3})$, ruling out inner product algorithms with time $O(N^{\frac{2}{3}-\eps})$.

We now present the construction of the vectors.
As a first step, we observe that for any set $X\subseteq\{1,...,U\}$,
we can compress its characteristic vector $v_{X}\in\{0,1\}^{U}$,
i.e., $v_{X}[i]=1$ iff $i\in X$, to size $O(|X|\log U)$ as follows.
We write $X=\{x_{1},\dots,x_{|X|}\}$ with $x_{1}<x_{2}<\dots<x_{|X|}$
and observe that
\[
v_{X}\coloneqq0^{x_{1}-1}\,1\,0^{x_{2}-x_{1}-1}\,1\dots1\,0^{x_{|X|}-x_{|X|-1}-1}\,1\,0^{U-x_{|X|}},
\]
where each 0-block has length at most $U$ and can thus be encoded
using $O(\log U$) symbols using Proposition~\ref{prop:SLPrepetition}. In total, we obtain a compression of size
$O(|X|\log U)$, which can be computed in time $O(|X|\log U)$ as
well.

Let $A=\{a_{1},\dots,a_{n}\}$. The central idea is to let $v'_{A+B},v'_{C}$
consist of $m$ blocks of size $2U$, where the $i$-th block in $v'_{A+B}$
gives the characteristic vector of the set $a_{i}+B=\{a_{i}+b\mid b\in B\}\subseteq\{1,\dots,2U\}$
and the $i$-th block in $v'_{C}$ gives the characteristic vector
of $C\subseteq\{1,\dots,2U\}.$ Formally, we define

\[
\begin{array}{ccccccc}
v_{A+B}' & \coloneqq & \underbrace{0^{a_{1}}v_{B}0^{U-a_{1}}}_{v_{a_{1}+B}'} & \underbrace{0^{a_{2}}v_{B}0^{U-a_{2}}}_{v_{a_{2}+B}'} & \dots & \underbrace{0^{a_{m}}v_{B}0^{U-a_{m}}}_{v_{a_{m}+B}'} & 0^{N-2mU},\\
v_{C}' & \coloneqq & v_{C}0^{U} & v_{C}0^{U} & \dots & v_{C}0^{U} & 0^{N-2mU}.
\end{array}
\]

(Here, the last block of $0$s only serves to get the desired dimension
of $N$ for technical reasons.) We observe that $v_{A+B}'$ and $v'_{C}$
have an inner product of at least $1$ if and only if the characteristic
vectors of some block $i$ have a common $1$-entry. Thus, consider
any block $i$: We have $v_{a_{i}+B}'[k]=(v_{C}0^{U})[k]=1$ if and
only if $k-a_{i}\in B$ and $k\in C$, i.e., $a_{i}\in A,k-a_{i}\in B,k\in C$
is a solution of the given 3SUM instance. Thus, $v_{A+B}'\cdot v_{C}'\ge1$
if and only if there is some $a\in A,b\in B,c\in C$ such that $a+b=c$,
as desired.

It remains to show that a $O(m\log U)$-sized compression of $v_{A+B}'$
and $v_{C}'$ can be computed in time $O(m\log U)$: Clearly, since
$v_{C}0^{U}$can be compressed to size $O(m\log U)$ efficiently,
we can also compress its $m$-fold repetition using $O(\log m)$ additional
symbols using Proposition~\ref{prop:SLPrepetition}, as well $0^{N-2mU}$ which takes $O(\log N)=O(\log mU)$
additional symbols; thus, $v_{C}'$ can be compressed to size $O(m\log mU)$
in time $O(m\log U)$. Furthermore, recall that we can compress $v_{B}$
to size $O(m\log U)$ efficiently, and let $\mathcal{G}$ be an SLP with starting symbol $S_B$ generating $v_B$. Thus, to compress $v_{a_i+B}$, we only need to compress the surrounding blocks $0^{a_i}$, $0^{U-a_i}$ and can reuse $S_B$ to generate $v_B$. Since we can encode the $0$-blocks using $O(\log U)$ additional non-terminals, this yields a compression size of $O(\log U)$ per block $i$. Together with a $O(\log mU)$ encoding
of the trailing block $0^{N-2mU}$, this yields again a compression of size $O(m\log U)$. Note that reusing a non-terminal generating $v_B$ was instrumental in giving a compression of size $O(m\log m)$ rather than $O(m^2\log m)$ and that this compression can indeed be computed in time $O(m\log U)$ and concludes the claim.
\end{proof}

With more work, the above arguments can be generalized to reduce a $k$SUM instance with $k$ sets of $m$ integers in $\{1,\dots, U\}$ to vectors of dimension $N=\Theta(m^{k-2}U)$ and compressed size $O(m\log U)$ in time $O(m \log U)$. 
The main idea is to encode a shift of $A_{k-1}$ for each tuple of $A_1,\ldots,A_{k-2}$ in one vector, and encode $m^{k-2}$ repetitions of the remaining set $A_k$ in the other vector.
Under the Strong $k$SUM conjecture, this yields a conditional lower bound for inner product of $\tilde{\Omega}(N^{1/3})$ where $n=O((N/U)^{1/(k-2)}\log N)$. Thus, for any fixed $\eps > 0$, let $k$ be a sufficiently large constant integer such that $1/(k-2) < \eps$, then the Strong $k$SUM conjecture implies that $N$-dimensional vectors with compressed size $n=O(N^\eps)$ cannot have an $O(N^{1/3-\delta})$ algorithm for any constant $\delta > 0$. We formally prove the result in the appendix.

\section{Matrix-Vector Multiplication}
\label{sec:MV}

In this section we sketch how to prove Theorem~\ref{thm:MV} by giving a reduction from 3SUM to Matrix-Vector multiplication on compressed data. We give a complete formal proof in the appendix.

A helpful tool for this task is the following self-reduction for 3SUM, which follows from combining a known self-reduction~\cite{LincolnWWW16} with a standard universe-size reduction technique on each produced instance~\cite{BDP05, Pat10,abboud2014losing}.  

\begin{lemma}[Self-Reduction for 3SUM]\label{lem:selfreduction}
Let $1\le s=s(m) \le m$ and $\epsilon > 0$ be arbitrary. If there is an algorithm that, given a target $t$ and $L=O( (m/s)^2 )$ sets $A_\ell,B_\ell,C_\ell$ of $s$ integers in $\{1,\dots,O(s^3\log^2 s)\}$, determines for all $1\le \ell \le L$ whether there are $a\in A_\ell, b\in B_\ell, c\in C_\ell$ with $a+b+c=t$ in total time $O(m^{2-\epsilon})$, then the 3SUM conjecture is false. 
\end{lemma}

Given the above self-reduction, the basic idea is as follows. We construct
a matrix $M$ whose rows are indexed by the instance $1\le\ell\le L$
and the aim is to construct the row $M_{\ell}$ and the vector $v$
such that $M_{\ell}\cdot v\ge1$ if and only if the instance $A_{\ell,}B_{\ell},C_{\ell}$
contains a solution, i.e., $a\in A_{\ell},b\in B_{\ell},c\in C_{\ell}$
with $a+b+c=t$. Unfortunately, we cannot apply our Vector Inner Product
construction directly: this would encode the set $A_{\ell}+B_{\ell}=\{a+b\mid a\in A_{\ell},b\in B_{\ell}\}$
into the row $M_{\ell}$ and the set $C_{\ell}$ into the vector $v$
-- however, in the matrix product $Mv$, each row $M_{\ell}$ is multiplied
with a \emph{fixed} vector $v$, while the $C_{\ell}$'s differ for
each $\ell$. We overcome this issue by adapting our construction
to encode the set $A_{\ell}+B_{\ell}+C_{\ell}=\{a+b+c\mid a\in A_{\ell},b\in B_{\ell},c\in C_{\ell}\}$
into the row $M_{\ell}$, and only the common target $t$ into $v$.
As all instances use the same target $t$, this is indeed possible.

Specifically, using the ideas of Theorem~\ref{thm:vectorInnerProduct}, which produces a $2sU$-dimensional vectors encoding the sets $A+B$ and $C$, both having compressed size $O(s\log U)$, we show how to produce $3s^2 U$-dimensional vectors $M_\ell$ and $v$ encoding the sets $A_\ell+B_\ell+C_\ell$ and $\{t\}$, both having compressed size $O(s\log U)$. This yields a $(L \times 3s^2 U)$-dimensional matrix $M$ and $3s^2 U$-dimensional vector $v$. There is a choice $s=\Theta(m^{2/7})$ that leads to a quadratic matrix $M$ with dimension $N=\Theta(m^{10/7})$ (as it has $O((m/s)^2)=O(m^{10/7})$ rows and $O(s^2U)=O(s^5)=O(m^{10/7})$ columns), with row compressions of size $n=\Theta(s\log s) = \Theta(m^{2/7}\log m)\approx N^{1/5}$. Thus, any $O(Nn^{2-\eps})$ algorithm computing $M\cdot v$ would solve 3SUM instances in time $\tilde{O}(m^{2-2\eps/7})$, refuting the 3SUM conjecture.

%\begin{theorem}
%\label{thm:MV}
%Assuming the 3SUM conjecture, no algorithm computes the product of
%an $N\times N$-dimensional matrix, where each row has compressed
%size $n=\Theta(N^{\frac{1}{5}})$, with an $N$-dimensional vector
%in time $O(Nn^{2-\epsilon})$ for some constant $\epsilon>0$.\end{theorem}

\section{Matrix-Matrix Multiplication}
\label{sec:MM}

In this section, we consider the problem of computing the matrix product
$C$ of two $N\times N$ matrices $A,B$. We consider the following
representations of the input matrices:
\begin{itemize}
\item \textbf{Convenient compression: }$A$ is compressed row-wise, $B$
is compressed column-wise. This representation allows us to compute
any single entry $C_{i,j}$ by running an inner product algorithm
on the compressed row $A_{i}$ and the compressed column $B_{j}$.
The size of the input is $O(N\bar{n}_{\mathrm{in}})$, where
$\bar{n}_{\mathrm{in}}$ is the maximum compressed size of the rows $A_{i}$ and columns $B_{j}$.
\item \textbf{Strong compression: } For any matrix $M$, we define \emph{strong compression} as a grammar compression of $M$ or $M^T$ when viewed as $n^2$-dimensional vector, whichever is shortest. When both $A,B$ are given as strong compression, the resulting representation can have a much smaller size (it can be $o(N)$),
but to compute a single entry $C_{i,j}$, we first might need to obtain
a representation of the row $A_{i}$ and the column $B_{j}$.
\end{itemize}
Similarly, we have several options for representing $C$:
\begin{itemize}
\item \textbf{Row-wise compression of $C$. }This compression is particularly
useful if we aim to compute repeated matrix products $A_{1}(A_{2}(\cdots(A_{k}B)))$.
The output size is $O(N\bar{n}_{\mathrm{out}})$, where $\bar{n}_{\mathrm{out}}$
is the maximum compressed size over all rows of $C$.
\item \textbf{Column-wise compression of $C$. }This compression is particularly
useful if we aim to compute repeated matrix products $(((AB_{1})B_{2})\cdots)B_{k}$.
The output size is $O(N\bar{n}_{\mathrm{out}})$, where $\bar{n}_{\mathrm{out}}$
is the maximum compressed size over all columns of $C$.
\item \textbf{Strong compression of $C$. }This compression has the smallest
output size, which can be even $o(N)$.
\end{itemize}
We show the following result:
\begin{theorem}
\label{thm:MMoutputLB}For infinitely many $N,$ there are $N\times N$
matrices $A,B$ with 
\begin{enumerate}
\item convenient compression of size $O(N\log N)$ (already under RLE),
and
\item strong compression of size $O(\log^{2}N)$, such that
\item the matrix product $C=AB$ has size $\Omega(N^{2}/\log^{2}N)$ in
any grammar-compression (row-wise, column-wise, or strong).
\end{enumerate}
\end{theorem}
As a consequence, there can be no $o(N^{2}/\log^2 N)$ algorithm for
matrix-matrix multiplication (for any of our discussed representations),
since already writing the output requires time $\Omega(N^{2}/\log^{2}N)$.

The rough proof strategy is to construct an instance $C=AB$ such that $C$ and $C^{T}$,
when viewed as $N^{2}$-dimensional vectors, contain all substrings of
length $2\log_{2}n$. By the following standard lemma, such a string
has no grammar compression of size $o(N^{2}/\log N)$.
\begin{lemma}[{see, e.g., \cite[Lemma 3]{Char+05}}]
\label{lem:incompressibility}Let $\ell\in\mathbb{N}$. If a string
$x$ is generated by a grammar of size~$n$, then $x$ contains at
most $n\ell$ distinct substrings of length $\ell$.\end{lemma}
%\mnote{state directly?
%If $x$ contains all $k$ distinct length-$\ell$ strings, then $n\ge k/\ell$...)}\end{lemma}
\begin{proof}[Proof of Theorem~\ref{thm:MMoutputLB}]
Let $\ell\in\mathbb{N}$.
 We first define the matrices $A',B'$ where $A'$ is a $(2^{\ell}\times2\ell)$
matrix with rows indexed by strings $x\in\{0,1\}^{\ell}$ in lexicographic
order, and $B'$ is a $(2\ell\times2^{\ell}(2\ell))$ matrix with
columns indexed by $(y,k)\in\{0,1\}^{\ell}\times\{1,\dots,2\ell\}$
in lexicographic order. For arbitrary $z\in\{0,1\}^{\ell}$, let $\diag(z)$
denote the $\ell\times\ell$ diagonal matrix with $z$ on the diagonal.
We define
\begin{align*}
A'_{x} & \coloneqq(x\mid1^{\ell}), & B'_{(y,1),\dots,(y,2\ell)} & \coloneqq\left(\begin{array}{c|c}
\diag(1^{\ell}) & 0\\
\hline 0 & \diag(y)
\end{array}\right).
\end{align*}

Let $C'=A'B'$ be the $(2^{\ell}\times2^{\ell}(2\ell))$ product matrix
of $A'$ and $B'$, with rows and columns indexed by $\{0,1\}^{\ell}$
and $\{0,1\}^{\ell}\times\{1,\dots,2\ell\}$, respectively. Observe
that by definition, $(C_{x,(y,1)},\dots,C_{x,(y,2\ell)})=(x\mid y)$
for any $x,y\in\{0,1\}^{\ell}$. In particular, when we view $C'$
as a $2^{2\ell}(2\ell)$-length string, it contains all strings in
$\{0,1\}^{2\ell}$ as substrings, thus by Lemma~\ref{lem:incompressibility},
any row-wise compression is of size at least $2^{2\ell}/(2\ell)$.

It is straightforward to make these matrices quadratic with dimension $N=\Theta(\ell2^\ell)$ (by introducing all-$0$ columns) and to ensure that also column-wise compression has size $\Omega(2^{2\ell}/\ell)=\Omega(N^2/\log^2 N)$ (using 
transposed constructions to $A'$ and $B'$). Finally, we can compress each row of $A'$ and column of $B'$ trivially to length $O(\ell)=O(\log N)$ (already using RLE). In the appendix, we also argue how to grammar-compress the concatenation of the columns of $A'$ and the rows of $B'$ to size $O(\ell^2)=O(\log^2 N)$, which concludes the desired bound on the strong compression.
\end{proof}

\section*{Broader Impact}

The broader impact of our work is to inform algorithm design for compressed linear algebra, which can lead to
faster algorithms for a variety of tasks on large data sets. The ethical consequences depend on the specific application. We do not see
any inherently new concerns raised by our results, beyond those that follow generally from faster algorithms and an increased ability to process data.

%Use unnumbered first level headings for the acknowledgments. All acknowledgments
%go at the end of the paper before the list of references. Moreover, you are required to declare 
%funding (financial activities supporting the submitted work) and competing interests (related financial activities outside the submitted work). 
%More information about this disclosure can be found at: \url{https://neurips.cc/Conferences/2020/PaperInformation/FundingDisclosure}.

%\paragraph{Acknowledgements.}

\bibliographystyle{plain}
\bibliography{refs.bib}

\appendix 

\section{Further Preliminaries}

For a sequence of vectors $v_1,\dots, v_\ell$, we let $v_1\,v_2\,\dots\,v_\ell = v_1\circ v_2 \circ \cdots \circ v_\ell = \bigcirc_{i=1}^\ell v_i$ denote their concatenation.

By the following observation, when proving a lower bound for a compression of size $\Theta(N^\gamma)$, the main task is to prove the upper bound $n = O(N^\gamma)$; the lower bound $n=\Omega(N^\gamma)$ can be ensured mechanically.
\begin{observation}\label{obs:padn}
Let $0 \le \gamma \le 1$. Given two $N$-dimensional vectors $u,v$ of compressed size $O(N^\gamma)$, we can compute two $O(N)$-dimensional vectors $u',v'$ of compressed size $\Theta(N^\gamma)$ with the same inner product.
\end{observation}
\begin{proof}
Append $0^{N^\gamma}$ using $\Theta(N^\gamma)$ additional rules to the encodings of $u$ and $v$.
\end{proof}

\paragraph{The Strong $k$SUM Assumption}
To generalize the lower bound of Theorem~\ref{thm:vectorInnerProduct} so that it works for an arbitrary relationship between compressed and uncompressed sizes, we will use an assumption about a generalized version of $3$SUM.

\begin{definition}[The $k$SUM Problem]
Given $k$ sets $A_1,\ldots,A_k$ of $m$ integers in $\{1,\ldots, U\}$, decide if there are $k$ numbers $a_1 \in A_1, \ldots, a_k \in A_k$ such that $a_1+\cdots+a_{k-1}=a_k$.
\end{definition}

For all constant $k \geq 3$ a simple meet-in-the-middle algorithm with hashing solves $k$SUM in $O(m^{\lceil k/2 \rceil})$ time, and no faster algorithm by $m^{\eps}$ factors, for any $\eps>0$, is known to date, unless the universe size $U$ is smaller than $O(m^{\lceil k/2 \rceil - \eps})$. This is because Fast Fourier Transform gives an $O(m+kU \log{U})$ time algorithm \cite{CLRSbook}.
It is conjectured that substantially faster algorithms do not exist (e.g. in \cite{abboud2013exact,ABBK17}).

\begin{ksumconj}
For all constant $k\geq 3$ it holds that: no algorithm can solve the $k$SUM problem with $U=O(m^{\lceil k/2 \rceil})$ in $O(m^{\lceil k/2 \rceil-\eps})$ time, where $\eps>0$.
\end{ksumconj}

Observe that this assumption is about all $k \geq 3$ and therefore implies the Strong 3SUM conjecture as a special case. Intuitively, the reason this problem helps us give reductions where the vectors are much more compressible is that, compared to 3SUM, as $k$ grows the ratio between the time complexity $m^{k/2}$ and the input size $m$ grows.

\section{Vector Inner Product}

In this section, we prove the generalization of the lower bound of Theorem~\ref{thm:vectorInnerProduct} to arbitrary relationships between compressed and uncompressed sizes of the vectors.

\begin{theorem}\label{thm:VVarbitraryn}
Let $0 < \eps < 1/3$. 
Assuming the Strong $k$SUM conjecture for all constant $k$, the inner product of two $N$-dimensional vectors that are grammar-compressed to size $n=\Theta(N^{\eps})$ cannot be computed in $O(N^{1/3-\delta})$ time, where $\delta > 0$.
\end{theorem}

This result follows from the following stronger statement.

\begin{theorem}\label{thm:vectorInnerProduct-generalk}
Let $k\ge 3$. 
Assuming the Strong $k$SUM conjecture, the inner product of two $N$-dimensional vectors that are grammar-compressed to size $n=\Theta(N^{1/\lceil \frac{3k-4}{2}\rceil})$ cannot be computed in $O(N^{(1/3+\gamma_k) - \delta})$ time, where $\delta > 0$ and
    \[ \gamma_k \coloneqq \begin{cases} \frac{2}{3(k-1)}, & \text{ if $k$ is odd,}\\ \frac{4}{9k-12}, & \text{ if $k$ is even.} \end{cases}
    \]
\end{theorem}

Observe that the above statement implies Theorem~\ref{thm:VVarbitraryn}: For any $0 < \eps < 1/3$, we choose $k$ sufficiently large such that $1/\lceil \frac{3k-4}{2}\rceil < \eps$. Then using Observation~\ref{obs:padn}, we obtain that any $O(N^{1/3-\delta})$-time algorithm for Vector Inner Product with compressed size $n=\Theta(N^\eps)$ would give an $O(N^{1/3+\gamma_k - \delta'})$-time algorithm for Vector Inner Product with compressed size $O(N^{1/\lceil \frac{3k-4}{2}\rceil})=O(N^\eps)$, where $\delta'=\gamma_k+\delta$ -- this would refute the Strong $k$SUM conjecture by Theorem~\ref{thm:vectorInnerProduct-generalk}.

Furthermore, observe that if we set $k=3$, we obtain a $\tilde{\Omega}(N^{2/3})$ lower bound for compressed size $n=\Theta(N^{1/3})$ under the Strong 3SUM conjecture.

In the remainder of this section, we give the proof of Theorem~\ref{thm:vectorInnerProduct-generalk}. The central construction is captured by the following lemma.

\begin{lemma}\label{lem:generalk-reduction}
	Given sets $A_{1},\dots,A_{k}$ of integers in $\{1,\dots,U\}$, we define
	%$A_{i}=\{a_{1}^{(i)},\dots,a_{m}^{(i)}\}$, we define
	\begin{align*}
		v_{A_{1}+\dots+A_{k-1}}' & \coloneqq\bigcirc_{\substack{(a_{1},\dots,a_{k-2})\in A_{1}\times\cdots\times A_{k-2}\\
				\text{in lexicographic order}
			}
		}0^{a_{1}+\dots+a_{k-2}}v_{A_{k-1}}0^{(k-2)U-a_{1}-\dots-a_{k-2}},\\
		v_{A_{k}}' & \coloneqq(v_{A_{k}}0^{(k-2)U})^{m^{k-2}},
	\end{align*}
	where $v_{A_{k-1}},v_{A_{k}}\in\{0,1\}^{U}$ denote the characteristic
	vectors of the sets $A_{k-1},A_{k}$. We have the following properties:
	\begin{enumerate}
		\item The inner product of the $m^{k-2}(k-1)U$-dimensional vectors $v_{A_{1}+\cdots+A_{k-1}}'$
		and $v_{A_{k}}'$ is nonzero if and only if there is a tuple $(a_{1},\dots,a_{k})\in A_{1}\times\cdots\times A_{k}$
		with $a_{1}+\dots+a_{k-1}=a_{k}$.
		\item We can compute compressions of $v_{A_{1}+\dots+A_{k-1}}',v_{A_{k}}'$
		of size $O(km\log U)=O(m\log U)$ in time $O(m\log U)$.
	\end{enumerate}
\end{lemma}
\begin{proof}
For 1., observe that by construction, $v'_{A_1+\dots+A_{k_1}}$ and $v'_{A_k}$ consist of $m^{k-2}$ blocks, indexed by $(a_1,\dots,a_{k-2})\in A_1\times \cdots \times A_{k-2}$ and consisting of the sequence $0^{a_1+\dots+a_{k-2}} v_{A_{k-1}} 0^{(k-2)U-a_1 -\cdots -a_{k-2}}$ and $v_{A_k} 0^{(k-2)U}$ of length $(k-1)U$, respectively. In particular, in block $(a_1,\dots,a_{k-2})$ there is a common $1$-entry $t$ if and only if $t = (a_1+a_2+\cdots+a_{k-2})+a$ for some $a\in A_{k-1}$ and $t=a'$ for some $a'\in A_{k}$. Thus, there exists a common $1$-entry in $v'_{A_1+\cdots+A_{k-2}}$ and $v'_{A_k}$ if and only if there are $(a_1,\dots,a_k)\in A_1\times \cdots \times A_{k}$ with $a_1+\dots+a_{k-1} = a_k$. 

For 2., we first recall that as shown in the proof of Theorem~\ref{thm:vectorInnerProduct}, we can compute a compression of the characteristic vectors $v_{A_{k-1}}$ and $v_{A_k}$ of size $O(m\log U)$ in time $O(m\log U)$. Thus, using Proposition~\ref{prop:SLPrepetition}, we can compute a compression of $v'_{A_{k}}= (v_{A_k}0^{(k-2)U})^{m^{k-2}}$ of size $O(m\log U) + O(\log((k-2)U)) + O(\log m^{k-2}) = O(m\log U)$ in time $O(m \log U)$. To show the claim for $v'_{A_1+\cdots+A_{k-1}}$, we proceed inductively and construct the strings $v'_{A_{k-1}} \coloneqq v_{A_{k-1}}$ and
\[ 	v_{A_{i}+\dots+A_{k-1}}' \coloneqq\bigcirc_{\substack{(a_{i},\dots,a_{k-2})\in A_{i}\times\cdots\times A_{k-2}\\
				\text{in lexicographic order}
			}
		}0^{a_{i}+\dots+a_{k-2}}v_{A_{k-1}}0^{(k-1-i)U-a_{i}-\dots-a_{k-2}},
\]
for $i=k-2,\dots,1$. The central observation is that we can write $A_i = \{a_1^{(i)},\dots,a_m^{(i)}\}$ with $a_1^{(i)} < a_2^{(i)} < \cdots < a_m^{(i)}$ and obtain
\[ v'_{A_i+\cdots+A_{k-1}} = \bigcirc_{j=1}^m 0^{a_j^{(i)}} \, v'_{A_{i+1}+\cdots+A_{k-1}} 0^{U-a_j^{(i)}}.\]
Thus, given an SLP $\mathcal{G}_{i+1}$ for $v'_{A_{i+1}+\cdots+A_{k-1}}$ with starting symbol $S_{i+1}$, we can give an SLP $\mathcal{G}_i$ for $v'_{A_{i}+\cdots+A_{k-1}}$ of size $|\mathcal{G}_{i+1}|+O(m\log U)$ as follows: For each $j=1,\dots,m$, we encode $0^{a^{(i)}_j}$ using $O(\log a^{(i)}_j)=O(\log U)$ additional symbols, re-use $S_{i+1}$ to generate $v'_{A_{i+1}+\cdots+A_{k-1}}$, and encode $0^{U-a^{(i)}_j}$ using $O(\log(U-a^{(i)}_j))=O(\log U)$ additional symbols. Observe that we can obtain this compression in time $O(m\log U)$.

Thus, starting from an SLP for $v'_{A_{k-1}}$, after $k-2$ steps we obtain an SLP $\mathcal{G}_1$ for $v_{A_{1}+\dots+A_{k-1}}'$ of size $O(km\log U)=O(m\log U)$. The running time of this construction is $O(km\log U)=O(m\log U)$, concluding the proof.
\end{proof}

Let $A_1,\dots,A_k \subseteq \{1,\dots,U\}$ be a Strong $k$SUM instance, i.e., $U=O(m^{\lceil k/2\rceil})$. The reduction given in Lemma~\ref{lem:generalk-reduction} gives two vectors $v,v'$ of dimension $m^{k-2}\cdot (k-1)U$ such that their inner product allows us to decide the $k$SUM instance. Furthermore, the vectors have a compressed size of $O(m\log U)$. 

We slightly adapt $v,v'$ by appending $0$'s to increase the dimension slightly to $N = m^{k-2}\cdot (k-1)U\log^{\lceil (3k-4)/2\rceil} U$ (this does not change their inner product). We verify the following facts: (1) an $O(N^{1/3+\gamma_k-\delta})$-time Vector Inner Product algorithm for some $\delta > 0$ refutes the Strong $k$SUM conjecture and (2) $n = O(N^{1/\lceil \frac{3k-4}{2}\rceil})$. Using Observation~\ref{obs:padn}, this concludes the proof of Theorem~\ref{thm:vectorInnerProduct-generalk}.

For (1), consider first the case that $k$ is odd. Then $U=O(m^{(k+1)/2})$ and $N=O(m^{k-2}U\mathrm{polylog} U)=O(m^{3(k-1)/2}\mathrm{polylog} m)$. Observe that 
\begin{align*}
N^{1/3+\gamma_k-\delta} & = O(m^{\frac{3(k-1)}{2}\cdot (\frac{1}{3}+\frac{2}{3(k-1)}-\delta)}\mathrm{polylog} m) \\
& = O(m^{\frac{k-1}{2}+1- \frac{3(k-1)}{2}\delta}) = O(m^{\lceil \frac{k}{2} \rceil-\delta'}),
\end{align*}
    for any  $0 < \delta'< 3(k-1)\delta/2$.
   
    Similarly, for even $k$, we have $U=O(m^{k/2})$ and $N=O(m^{k-2}U\mathrm{polylog} U)=O(m^{(3k-4)/2}\mathrm{polylog} m)$. Using $1/3+\gamma_k = 1/3+4/(9k-12)= k/(3k-4)$, we obtain that
\[ N^{1/3+\gamma_k-\delta} = O(m^{\frac{3k-4}{2}\cdot (\frac{k}{3k-4}-\delta)}\mathrm{polylog} m) = O(m^{\frac{k}{2}- \delta'}),\]
for any $0< \delta'<(3k-4)\delta/2$. Thus, in both cases, an $O(N^{1/3+\gamma_k-\delta})$-time Vector Inner Product algorithm refutes the Strong $k$SUM conjecture by solving the given $k$SUM instance in time $O(m^{\lceil k/2 \rceil -\delta'})$ with $\delta'>0$. 

Finally, for (2), note that $N=O(m^{k-2}U\log^{\lceil (3k-4)/2 \rceil} U)=O(m^{\lceil (3k-4)/2\rceil}\log^{\lceil (3k-4)/2\rceil}m)$. Thus $n=O(m\log m)=O(N^{1/\lceil (3k-4)/2\rceil})$, as desired.

\section{Matrix-Vector Product}

In this section we provide the full proof of Theorem~\ref{thm:MV}.
We first prove a self-reduction for 3SUM as a central tool (using standard techniques), and then proceed to give the final reduction.

\subsection{Proof of the Self-Reduction}

Let us restate Lemma~\ref{lem:selfreduction}.

\begin{lemma}[Self-Reduction for 3SUM]\label{lem:selfreduction-appendix}
Let $1\le s=s(m) \le m$ and $\eps> 0$ be arbitrary. If there is an algorithm that, given a target $t$ and $L=O( (m/s)^2 )$ sets $A_\ell,B_\ell,C_\ell$ of $s$ integers in $\{1,\dots,O(s^3\log^2 s)\}$, determines for all $1\le \ell \le L$ whether there are $a\in A_\ell, b\in B_\ell, c\in C_\ell$ with $a+b+c=t$ in total time $O(m^{2-\epsilon})$, then the 3SUM conjecture is false. 
\end{lemma}

In the remainder of this section, we give the proof.

Let $A,B,C$ be sets of $m$ integers in $\{1,\dots,U\}$. We use
a couple of results from earlier work that are stated for the following
3SUM formulation: given three sets $A',B',C'$ of $m$ integers in
$\{-U,\dots,U\}$ with $U=O(m^{3}\log^2 m)$, we are asked to determine
whether there are $a\in A',b\in B',c\in C'$ such that $a+b+c=0$.
We first reduce our formulation to this formulation by setting $A'\coloneqq A,B'\coloneqq B,$
and $C'\coloneqq-C=\{-c\mid c\in C\}$. We can now use the following
known self-reduction for 3SUM.
\begin{lemma}[Reformulated from {\cite[Theorem 13]{LincolnWWW16}}]\label{lem:selfreduction-core}
Let $s\coloneqq s(m)$ with $1\le s\le m$.
Given three sets $A',B',C'$ of $m$ integers in $\{-U,\dots,U\}$,
we can compute, in time $O(m^{2}/s)$, a list of $L=O((m/s)^{2})$
3SUM instances, i.e., sets $A'_{\ell},B_{\ell}',C_{\ell}'$ with $1\le\ell\le L$,
such that there is an $a\in A',b\in B',c\in C'$ with $a+b+c=0$ if
and only if there is an instance $1\le\ell\le L$ and a triple $a\in A'_{\ell},b\in B_{\ell}',c\in C_{\ell}'$
with $a+b+c=0$. Furthermore, each $A'_{\ell},B'_{\ell},C'_{\ell}$
is a subset of $s$ integers of $A',B',C'$, respectively.
\end{lemma}
\begin{proof}[Proof sketch]
 We give the high-level arguments (for details, see the proof of Theorem 13 in~\cite{LincolnWWW16}). For a set
$S$, let $\min S$ and $\max S$ denote the smallest and largest
element in $S$, respectively. We sort $A',B',C'$ and split each
array into $\lceil m/s\rceil$ consecutive parts $A'_{1},\dots,A'_{\lceil m/s\rceil},B'_{1},\dots,B_{\lceil m/s\rceil}',C_{1}',\dots,C_{\lceil m/s\rceil}'$,
each of at most $s$ elements, such that $\max A_{i}'<\min A_{i+1}'$,$\max B_{i}'<\min B_{i+1}'$
and $\max C_{i}'<\min C_{i+1}'$ for all~$i$. Instead of searching
for a 3SUM triple $a\in A_{i}',b\in B_{j}',c\in C_{k}'$ for each
$1\le i,j,k\le\lceil m/s\rceil$ (i.e., $\Theta((m/s)^{3})$ subproblems
with $s$ elements each), one observes that most subproblems can be
trivially solved: We say that a subproblem $(i,j,k)$ is trivial,
if $\min A_{i}+\min B_{j}+\min C_{k}>0$ or $\max A_{i}+\max B_{j}+\max C_{k}<0$;
these subproblems cannot contain a solution. The key insight is that
there are at most $O((m/s)^{2})$ non-trivial subproblems (which follows
since the \emph{domination} partial ordering on $\{1,\dots,u\}^{3}$
has at most $O(u^{2})$ incomparable elements); these can be determined
in time $O((m/s)^{2})$. Thus, it suffices to list all $O((m/s)^{2})$
non-trivial subproblems with $s$ integers in each set in time $O(m^{2}/s)$.
\end{proof}
The resulting instances $A_{\ell}',B_{\ell}',C_{\ell}'$ consist of
integers in $\{-U,\dots,U\}$ with large universe size $U=O(m^{3}\log^{2}m)$.
We reduce the universe size to $O(s^{3}\log^{2}s)$ using a folklore
technique (a slightly stronger result with $U=O(s^{3})$ can be achieved
using the techniques of \cite{BDP05}). To prepare notation, for any
set $S$, we let $S\bmod p\coloneqq\{s\bmod p\mid s\in S\}$.

\begin{lemma}[Adaptation of {\cite[Lemma B.1]{abboud2014losing}}]\label{lem:universe-reduction}
 There is some $\alpha$ such that $U'\coloneqq\alpha s^{3}\log s\log U$
satisfies the following property: Let $A,B,C$ be sets of $s$ integers
in $\{-U,\dots,U\}$ such that no $a\in A,b\in B,c\in C$ satisfies
$a+b+c=0$. Let $p$ be a prime chosen uniformly at random from $\{2,\dots,U'\}$.
Then the probability that there are $a_{p}\in A\bmod p,b_{p}\in B\bmod p,c_{p}\in C\bmod p$
with $a_{p}+b_{p}+c_{p}\equiv0\pmod p$ is at most $1/2$.
\end{lemma}
\begin{proof}
Let $a\in A,b\in B,c\in C$ be arbitrary. Since $a+b+c\ne0$, note
that $(a\bmod p)+(b\bmod p)+(c\bmod p)\equiv0\pmod p$ if and only
if $p$ divides $a+b+c$. Since $a+b+c\in\{-3U,\dots,3U\}$, $a+b+c$
has at most $\log_{2}(3U)$ prime factors. Let $P$ denote the number of
prime numbers in $\{2,\dots,U'\}$; by the prime number theorem we can
choose $\alpha$ large enough such that $P\ge2s^{3}\log_{2}(3U)$. Thus,
the probability that $p$ was chosen among these at most $\log_{2}(3U)$ prime
factors is at most $\log_{2}(3U)/P\le1/(2s^{3})$. Thus, by a union
bound over all $s^{3}$ triples $a\in A,b\in B,c\in C$, the probability
that there are $a_{p}\in A\bmod p,b_{p}\in B\bmod p,c_{p}\in C\bmod p$
with $a+b+c\equiv0\pmod p$ is at most $1/2$.
\end{proof}
Note that if $A,B,C$ contain a triple $a,b,c$ with $a+b+c=0$, then
also $A\bmod p,B\bmod p,C\bmod p$ contain a triple $a_{p},b_{p},c_{p}$
with $a_{p}+b_{p}+c_{p}\equiv0\pmod p$ for any $p$.

We can finally prove Lemma~\ref{lem:selfreduction-appendix}: Assume that there is an
algorithm $\mathcal{A}$ that given a target $t$ and $L=O((m/s)^{2})$
instances $A_{\ell},B_{\ell},C_{\ell},1\le\ell\le L$ of $s$ integers
in $\{1,\dots,U'\}$, determines for all $1\le\ell\le L$ whether
there are $a\in A_{\ell},b\in B_{\ell},c\in C_{\ell}$ with $a+b+c=t$
in total time $O(m^{2-\eps})$ with $\eps>0$. Observe that since $\mathcal{A}$ runs in time $O(m^{2-\eps})$, we must have $s=\Omega(m^\eps)$, since otherwise already the size of the input to $\mathcal{A}$ of $\Theta(m^2/s)$ would be $\omega(m^{2-\eps})$. Thus, we have $U'=O(s^3\log^2 s)$.

For $r=1,\dots,\gamma\log m$
many repetitions, we do the following: We choose a random prime $p_{r}\in[2,U']$
and obtain $\ell$ instances in $\{0,\dots,p_{r}-1\}\subseteq\{0,\dots,U\}$
by taking the sets modulo $p_{r}$, i.e., $A_{\ell}^{(r)}\coloneqq A'_{\ell}\bmod p_{r},$
$B_{\ell}^{(r)}\coloneqq B'_{\ell}\bmod p_{r}$, and $C_{\ell}^{(r)}=C'_{\ell}\bmod p_{r}$.
Observe that we may determine whether there is some $a\in A_{\ell}^{(r)},b\in B_{\ell}^{(r)},c\in C_{\ell}^{(r)}$
with $a+b+c\equiv0\pmod {p_{r}}$ by testing for each $t\in\{0,p_{r},2p_{r}\}$,
whether there $a\in A_{\ell}^{(r)},b\in B_{\ell}^{(r)},c\in C_{\ell}^{(r)}$with
$a+b+c=t$. Thus, to do this, and additionally ensure that each integer
is in $\{1,\dots,U'\}$, we add $1$ to each integer in $A_{\ell}^{(r)},B_{\ell}^{(r)},C_{\ell}^{(r)}$
and for each $\lambda\in\{0,1,2\}$, call $\mathcal{A}$ on the sets
$A_{\ell}^{(r)},B_{\ell}^{(r)},C_{\ell}^{(r)},1\le\ell\le L$ with
common target $t_{\lambda}\coloneqq3+\lambda p_{r}$.

Observe that after these $3\gamma\log m$ calls to $\mathcal{A}$,
we know for each $1\le\ell\le L$ and $1\le r\le\gamma\log m$ whether
there are $a\in A_{\ell}',b\in B_{\ell}',c\in C{}_{\ell}'$ with $a+b+c\equiv0\pmod {p_{r}}$.
We declare our original 3SUM instance $A,B,C$ to be a YES instance
if and only if there is some $\ell$ such that for all $r$ we have
found a witness $a\in A_{\ell}',b\in B_{\ell}',c\in C{}_{\ell}'$
with $a+b+c\equiv0\pmod {p_{r}}$. Note that if $A,B,C$ is a YES instance,
we always return YES by Lemma~\ref{lem:selfreduction-core}. Otherwise, if $A,B,C$
is a NO instance, consider a fixed $\ell.$ By Lemmas~\ref{lem:selfreduction-core} and~\ref{lem:universe-reduction}, the probability
that for all $r$, we find $a\in A_{\ell}',b\in B_{\ell}',c\in C_{\ell}'$
with $a+b+c\equiv0\pmod {p_{r}}$ is bounded by $2^{-\gamma\log m}=m^{-\gamma}$.
Thus, by a union bound over all $\ell$, the probability that we incorrectly
return YES in this case is at most $Lm^{-\gamma}=O((m/s)^{2}m^{-\gamma})=O(m^{2-\gamma})$.
We can make this error probability polynomially small by choosing
$\gamma>2$. 

Observe that the running time of the above process is $O(\log m)$ times the running time of $\mathcal{A}$ (note that the running time used for Lemma~\ref{lem:selfreduction-core} is linear in its output size, which is the input size of $\mathcal{A}$ and thus dominated by the running time of $\mathcal{A}$). Thus, we can solve any 3SUM instance in time $O(m^{2-\eps}\log m)$, which would refute the 3SUM conjecture. This concludes the proof of Lemma~\ref{lem:selfreduction-appendix}.

%\begin{proof}[Proof sketch]
%We use a detour via the (subquadratically equivalent) Convolution 3SUM problem: Using Patrascu's/Kopelowitz et al., we obtain an equivalent instance of Convolution 3SUM. We divide each input sequence into $\lceil n/s\rceil$ contiguous parts of size at most $s$, and note that there are at most $2\lceil n/s\rceil^2$ many subproblems involving only parts of size at most $s$. Any subproblem we may view as a 3SUM instance using the straightforward reduction, and may use a standard hashing to reduce the universe size to $100s^3$; specifically, there are $O(\log n)$ hash functions such that for any YES instance, all hashed down instances are YES instances, and for any NO instance, all hashed down instances are YES instances w.p. at most $n^{-\alpha}$.\mnote{hmm, ``all'' is not quite true because of \emph{almost}-linearity}
%\end{proof}

\subsection{Main Reduction for Matrix-Vector Multiplication}

We now turn to the proof of Theorem~\ref{thm:MV}.

\begin{proof}
Let $s$ be a parameter to be chosen later. By Lemma~\ref{lem:selfreduction}, it suffices
to solve $L=O((m/s)^{2})$ 3SUM instances $A_{\ell},B_{\ell},C_{\ell}$
consisting of $s$ integers in $\{1,\dots,U\},U=O(s^3\log^2 s)$ with common
target $1\le t\le3U$ in time $O(m^{2-\epsilon})$ for some $\epsilon > 0$ to contradict the
3SUM conjecture.

We construct an $(L\times3s^{2}U)$ matrix $M$ and $v\in\{0,1\}^{3s^{2}U}$
as follows. Intuitively, each row $M_{\ell}$ and the vector $v$
are partitioned into $s^{2}$ blocks of size $3U$. Each block is
indexed by $(i,j)$ with $i,j\in\{1,\dots,s\}$ in lexicographic order
and the block of $M_{\ell}$ corresponding to $(i,j)$ encodes the
characteristic vector of the set $a_{i}+b_{j}+C_{\ell}=\{a_{i}+b_{j}+c\mid c\in C_{\ell}\}\subseteq\{1,\dots,3U\},$
where $a_{i}$ is the $i$-th integer in $A_{\ell}$ and $b_{j}$
is the $j$-th integer in $B_{\ell}$. Correspondingly, every block
$(i,j)$ in $v$ encodes the characteristic vector of the singleton
set $\{t\}\subseteq\{1,\dots,3U\}$. Thus, there is a position in block $(i,j)$
in which both $M_{\ell}$ and $v$ have a $1$ if and only if there is
a $c\in C_{\ell}$ such that $a_{i}+b_{j}+c=t$.

Formally, for any $1\le\ell\le L$, we write $A_{\ell}=\{a_{1}^{\ell},\dots,a_{s}^{\ell}\},B_{\ell}=\{b_{1}^{\ell},\dots,b_{s}^{\ell}\}$
and define

\[
\begin{array}{ccccccc}
M_{\ell} & \coloneqq & \underbrace{0^{a_{1}+b_{1}}v_{C_{\ell}}0^{3U-a_{1}-a_{2}}}_{v_{a_{1}^{\ell}+b_{1}^{\ell}+C_{\ell}}} & \dots & \underbrace{0^{a_{i}+b_{j}}v_{C_{\ell}}0^{3U-a_{i}-b_{j}}}_{v_{a_{i}^{\ell}+b_{j}^{\ell}+C_{\ell}}} & \dots & \underbrace{0^{a_{s}+b_{s}}v_{C_{\ell}}0^{3U-a_{s}-b_{s}}}_{v_{a_{s}^{\ell}+b_{s}^{\ell}+C_{\ell}}},\\
v & \coloneqq & 0^{t-1}10^{3U-t} & \dots & 0^{t-1}10^{3U-t} & \dots & 0^{t-1}10^{3U-t},
\end{array}
\]
where $v_{C_{\ell}}\in\{0,1\}^{U}$ denotes the characteristic vector
of $C_{\ell}$. By this structure, it is clear that $M_{\ell}v\ge1$
if and only if there are $a\in A_{\ell},b\in B_{\ell},c\in C_{\ell}$
with $a+b+c=t$.

We will show that each row $M_{\ell}$ can be compressed to size $\Theta(s\log s)$
(as opposed to its RLE of length $\Theta(s^{3}\log s)$). We thus
will set $N=\lceil 3s^{2}U\log^{3}s \rceil=\Theta(s^{5}\log^{5}s)$, and append
$0^{N-3s^{2}U}$ to each row $M_{\ell}$ and $v$, so that we obtain
an $L\times N$ matrix $M'$ and $N$-dimensional vector $v'$ whose
product $M'v'$ can be used to solve all instances $A_{\ell},B_{\ell},C_{\ell}$
in linear time. Observe that each row has a compression of size $\Theta(N^{1/5})=\Theta(s\log s)$,
as desired. Since $L=O((m/s)^{2})$ and $N\ge s^{5}$, we can set
$s=\Theta(m^{2/7})$ such that $L\le N$ (we can indeed make $L=N$
by introducing zero rows, if necessary). Thus, an $O(Nn^{2-\epsilon})$-time
algorithm for multiplying $M'$ and $v'$ would solve all $L$ 3SUM
instances in time 
\[
O(Nn^{2-\epsilon})=O((m/s)^{2}(s\log s)^{2-\epsilon})=O((m^{2}/s^{\epsilon})\mathrm{polylog}s)=O(m^{2-\frac{2}{7}\epsilon}\mathrm{polylog}m),
\]
which would refute the 3SUM conjecture.

Analogous to the proof of Theorems~\ref{thm:vectorInnerProduct} and~\ref{thm:vectorInnerProduct-generalk}, we can compute a compression of size $\Theta(s \log s)$ in time $O(s\log s)$. Indeed, for each $M_\ell$, this already follows from Lemma~\ref{lem:generalk-reduction} when setting $A_1 \coloneqq A_\ell, A_2 \coloneqq B_\ell, A_3 \coloneqq C_\ell$, which shows how to compress the string $v'_{A_1+A_2+A_3} = M_\ell$ to size $O(s\log U)=O(s\log s)$ in time $O(s\log U)=O(s\log s)$. For $v$, we simply apply Proposition~\ref{prop:SLPrepetition} to the straightforward compression of $0^{t-1}1 0^{3U-t}$ to size $O(\log U)$, which leads to a compression of $v$ of size $O(\log U + \log s) = O(\log s)$. Using Observation~\ref{obs:padn}, we can make all encodings have size $\Theta(s\log s)$, which concludes the proof.
\end{proof}

\section{Matrix-Matrix Product}

In this section, we give the full proof of Theorem~\ref{thm:MMoutputLB}.

\begin{proof}[Proof of Theorem~\ref{thm:MMoutputLB}]
Let $\ell\in\mathbb{N}$.
 We first define the matrices $A',B'$ where $A'$ is a $(2^{\ell}\times2\ell)$
matrix with rows indexed by strings $x\in\{0,1\}^{\ell}$ in lexicographic
order, and $B'$ is a $(2\ell\times2^{\ell}(2\ell))$ matrix with
columns indexed by $(y,k)\in\{0,1\}^{\ell}\times\{1,\dots,2\ell\}$
in lexicographic order. For arbitrary $z\in\{0,1\}^{\ell}$, let $\diag(z)$
denote the $\ell\times\ell$ diagonal matrix with $z$ on the diagonal.
We define
\begin{align*}
A'_{x} & \coloneqq(x\mid1^{\ell}), & B'_{(y,1),\dots,(y,2\ell)} & \coloneqq\left(\begin{array}{c|c}
\diag(1^{\ell}) & 0\\
\hline 0 & \diag(y)
\end{array}\right).
\end{align*}

Let $C'=A'B'$ be the $(2^{\ell}\times2^{\ell}(2\ell))$ product matrix
of $A'$ and $B'$, with rows and columns indexed by $\{0,1\}^{\ell}$
and $\{0,1\}^{\ell}\times\{1,\dots,2\ell\}$, respectively. Observe
that by definition, $(C_{x,(y,1)},\dots,C_{x,(y,2\ell)})=(x\mid y)$
for any $x,y\in\{0,1\}^{\ell}$. In particular, when we view $C'$
as a $2^{2\ell}(2\ell)$-length string, it contains all strings in
$\{0,1\}^{2\ell}$ as substrings, thus by Lemma~\ref{lem:incompressibility},
any row-wise compression is of size at least $2^{2\ell}/(2\ell)$.

To also ensure column-wise incompressibility, we slightly extend the
construction by analogous transposed constructions: We let $N\coloneqq2^{\ell}(2\ell+1)$
and define the final $(N\times N)$ matrices $A,B$ as follows:
\begin{align*}
A & \coloneqq\left(\begin{array}{c|c|c}
A' & 0 & 0\\
\hline 0 & B'^{T} & 0
\end{array}\right), & B & \coloneqq\left(\begin{array}{c|c}
B' & 0\\
\hline 0 & A'^{T}\\
\hline 0 & 0
\end{array}\right).
\end{align*}
Since $C\coloneqq AB=\left(\begin{array}{c|c}
A'B' & 0\\
\hline 0 & (A'B')^{T}
\end{array}\right)$ contains all length-$(2\ell)$ strings as substrings of the rows
(in the $A'B'$ part) and as substrings of the columns (in the $(A'B')^{T}$
part), any strong compression of $C$ is of size at least $2^{2\ell}/(2\ell)=\Omega(N/\log^{2}N)$,
proving the third part of the claim.

For the first two parts, it remains to show that $A$ and $B$ can
be well compressed: For the convenient compression, we observe that
any row in $A$ is either of the form $(x1^{\ell}\mid0^{2\ell}\mid0^{N-4\ell})$,
which has a RLE of length at most $|x1^{\ell}|+O(\log N)=O(\log N)$,
or it is of the form $(0^{2\ell}\mid0^{i-1}\alpha0^{2\ell-i}\mid0^{N-4\ell})$
for some $\alpha\in\{0,1\},i\in\{1,...,2\ell\}$, which also has a
RLE of length at most $O(\log N)$. Thus, each of the $N$ rows of
$A$ can be compressed to size $O(\log N)$, as desired. By a symmetric
statement, also each column of $B$ has a RLE of size $O(\log N)$.

Finally, for the strong compression, we show that we compress $A^{T}$
when viewed as a string, i.e., we compress the concatenation of the
columns of $A$. The main insight is the following: Imagine a binary
$\ell$-bit counter. Using grammar compression, we can compress the
sequence of values of any fixed bit while the counter counts from
$0$ to $2^{\ell}-1$ in size $O(\ell).$ Formally, let $G_{0},G_{1}$
be grammar compressions of strings $s_{0}$,$s_{1}$. For any $1\le i\le\ell$,
we can encode $(s_{0}^{2^{\ell-i}}s_{1}^{2^{\ell-i}})^{2^{i-1}}$
using only $O(\ell)$ additional non-terminals in the canonical way.
Specifically, using $O(\ell-i)$ new symbols, we may encode $s_{0}^{2^{\ell-i}}s_{1}^{2^{\ell-i}}$;
let $\tilde{S}$ denote the corresponding non-terminal. We then encode
$\tilde{S}^{2^{i-1}}$ using $O(i)$ additional new symbols. In total,
we only need $O((\ell-i)+i)=O(\ell)$ additional symbols, as desired.

We apply the above idea to encode the concatenation all columns of
$A$ as follows: Consider column~$i$.
\begin{itemize}
\item For $1\le i\le\ell$, then by the chosen lexicographic order of the
row indices $x\in\{0,1\}^{\ell}$ of $A'$, note that the $i$-th
column of $A$ is of the form $(0^{2^{\ell-i}}1^{2^{\ell-i}})^{2^{i-1}}\mid0^{N-2^{\ell}}$.
Using the above analysis, we can compress it to size $O(\ell)+O(\log N)=O(\log N)$. 
\item If $\ell+1\le i\le2\ell$, the $i$-th column is of the form $1^{2^{\ell}}\mid0^{N-2^{\ell}}$,
which we can compress to size $O(\log\ell+\log N)=O(\log N)$. 
\item If $2\ell+1\le i\le3\ell$, write $i=2\ell+i'$ and observe that the
$i$-th column of $A$ is of the form $0^{2^{\ell}}\mid(0^{i'-1}10^{\ell-i'})^{2^{\ell}}$.
Using $O(\ell)$ non-terminals to encode $0^{i'-1}10^{\ell-i'}$,
it is immediate that we can compress the complete column using $O(\ell)$
additional non-terminals, i.e., yielding a total of $O(\ell)=O(\log N)$.
\item If $3\ell+1\le i\le4\ell,$ write $i=3\ell+i'$ and observe that by
the chosen lexicographic order of the column indices $(y,k)\in\{0,1\}^{\ell}\times\{1,\dots,2\ell\}$ of
$B'$, the $i$-th column of $A$ is of the form $0^{2^{\ell}}\mid(s_{0}^{2^{\ell-i'}}s_{1}^{2^{\ell-i'}})^{2^{i'-1}}$
where $s_{\alpha}\coloneqq0^{i'-1}\alpha1^{\ell-i'}$. We can give
trivial grammars of size $O(\ell)$ for $s_{0},s_{1}$. Then, by the
above analysis, we only need $O(\ell)$ additional non-terminals for
the counter-like part. In total, we only need $O(\ell)=O(\log N)$
non-terminals to encode the $i$-th column.
\item Finally, observe that the remaining columns $i=4\ell+1,\dots,N$ consist
of $(N-4\ell)N$ zeroes, which we can encode together using only $O(\log N)$
non-terminals.
\end{itemize}

In summary, we can encode the first $4\ell$ columns using $O(\log N)$
non-terminals each, and only $O(\log N)$ non-terminals for the remaining
columns, so we can fully compress the concatenation of $A$'s columns
to size $O(\log^{2}N)$, as claimed.\end{proof}

%\newpage
%\appendix

\end{document}